\documentclass[sigconf,natbib=true,screen]{acmart} 

\usepackage{graphicx,subfigure}
\usepackage{multirow}
\usepackage{balance} 

\usepackage{amsmath}
\usepackage{mathtools}
\usepackage{amsthm}

\usepackage{xfunctions}
\usepackage[ruled,linesnumbered]{algorithm2e}
\usepackage{xcolor}
\usepackage{nicefrac}

\input{aliases.tex}

\newtheorem{remark}{Remark}
\newtheorem{theorem}{Theorem}
\newtheorem{proposition}[theorem]{Proposition}

\AtBeginDocument{%
  \providecommand\BibTeX{{%
    \normalfont B\kern-0.5em{\scshape i\kern-0.25em b}\kern-0.8em\TeX}}}

\setcopyright{acmcopyright}

\copyrightyear{2022}
\acmYear{2022}
\setcopyright{acmlicensed}\acmConference[SIGIR '22]{Proceedings of the 45th
International ACM SIGIR Conference on Research and Development in
Information Retrieval}{July 11--15, 2022}{Madrid, Spain}
\acmBooktitle{Proceedings of the 45th International ACM SIGIR Conference on
Research and Development in Information Retrieval (SIGIR '22), July 11--15,
2022, Madrid, Spain}
\acmPrice{15.00}
\acmDOI{10.1145/3477495.3532035}
\acmISBN{978-1-4503-8732-3/22/07}

\begin{document}
\fancyhead{}

\title{Optimizing Generalized Gini Indices for Fairness in Rankings}

\author{Virginie Do}
\email{virginiedo@fb.com}
\affiliation{
  \institution{Meta AI Research}
  \country{}
}
\affiliation{
  \institution{LAMSADE, Université Paris Dauphine-PSL}
  \country{}
}

\author{Nicolas Usunier}
\email{usunier@fb.com}
\affiliation{
  \institution{Meta AI Research}
  \country{}
}

\renewcommand{\shortauthors}{Do and Usunier}

\begin{abstract}

There is growing interest in designing recommender systems that aim at being fair towards item producers or their least satisfied users. Inspired by the domain of inequality measurement in economics, this paper explores the use of generalized Gini welfare functions (GGFs) as a means to specify the normative criterion that recommender systems should optimize for. GGFs weight individuals depending on their ranks in the population, giving more weight to worse-off individuals to promote equality. Depending on these weights, GGFs minimize the Gini index of item exposure to promote equality between items, or focus on the performance on specific quantiles of least satisfied users.
GGFs for ranking are challenging to optimize because they are non-differentiable. We resolve this challenge by leveraging tools from non-smooth optimization and projection operators used in differentiable sorting.
We present experiments using real datasets with up to $15$k users and items, 
which show that our approach obtains better trade-offs than the baselines on a variety of recommendation tasks and fairness criteria.
\end{abstract}

\ccsdesc[500]{Information systems~Retrieval models and ranking}


\keywords{fairness, ranking, recommender systems, welfare economics}

\maketitle

\section{Introduction} \label{sec:intro}

Recommender systems play an important role in organizing the information available to us, by deciding which content should be exposed to users and how it should be prioritized. These decisions impact both the users and the item producers of the platform. While recommender systems are usually designed to maximize performance metrics of user satisfaction, several audits recently revealed potential performance disparities across users \citep{sweeney2013discrimination,datta2015automated,ekstrand2018all,mehrotra2017auditing}. On the side of item producers, the growing literature on fairness of exposure aims to avoid popularity biases \citep{abdollahpouri2019unfairness} by reducing inequalities in the exposure of different items \citep{singh2018fairness}, or aiming for equal exposure weighted by relevance \citep{diaz2020evaluating,biega2018equity,morik2020controlling}. In most cases, the approaches proposed for user- and item-side fairness aim to reduce inequalities.

In this paper, we propose a new approach to fair ranking based on Generalized Gini welfare Functions (GGFs, \citep{weymark1981generalized}) from the economic literature on inequality measurement \citep{cowell2000measurement}. GGFs are used to make decisions by maximizing a weighted sum of the utilities of individuals which gives more weight to those with lower utilities. By prioritizing the worse-off, GGFs promote more equality. 

The normative appeal of GGFs lies in their ability to address a multiplicity of fairness criteria studied in the fair recommendation literature. Since GGFs include the well-known Gini inequality index as a special case \citep{gini1921measurement}, they can be used to optimize trade-offs between exposure inequality among items and user utility, a goal seeked by many authors \citep{morik2020controlling,zehlike2020reducing}. GGFs also conveniently specify normative criteria based on utility quantiles \citep{do2021two}: for instance, it is possible to improve the utility of the $10\%$ worse-off users and/or items with GGFs, simply by assigning them more weight in the objective. Moreover, using techniques from convex multi-objective optimization, we show that GGFs cover \emph{all} ranking policies that satisfy \emph{Lorenz efficiency}, a distributive justice criterion which was recently introduced for two-sided fairness in rankings \citep{do2021two}.


The difficulty of using GGFs as objective functions for fairness in ranking stems from their non-differentiability, which leads to computational challenges. Indeed, ranking with fairness of exposure requires the solution of a global optimization problem in the space of (randomized) rankings of all users, because the exposure of an item is the sum of its exposure to every users. The Frank-Wolfe algorithm \citep{frank1956algorithm} was shown to be a computationally efficient method for maximizing globally fair ranking objectives, requiring only one top-$K$ sort operation per user at each iteration \citep{do2021two}. However, vanilla Frank-Wolfe algorithms only apply to objective functions that are differentiable, which is not the case of GGFs.


We propose a new algorithm for the optimization of GGFs based on extensions of Frank-Wolfe algorithms for non-smooth optimization \citep{lan2013complexity,yurtsever2018conditional,2020thekumparampil}. These methods usually optimize smoothed surrogate objective functions, while gradually decreasing a smoothing parameter, and a common smoothing technique uses the Moreau envelope \citep{moreau1962fonctions,yosida1965functional}. Our main insight is that the gradient of the Moreau envelope of GGFs can be computed in $O(n \log n)$ operations, where $n$ is the number of users or items. This result unlocks the use of Frank-Wolfe algorithms with GGFs, allowing us to efficiently find optimal ranking policies while optimizing GGFs. 

We showcase the performances of the algorithm on two recommendations tasks of movies and music, and on a reciprocal recommendation problem (akin to dating platforms, where users are recommended to other users), with datasets involving up to $15k$ users and items. Compared to relevant baselines, we show that our algorithm successfully yields better trade-offs in terms of user utility and inequality in item exposure measured by the Gini index. Our approach also successfully finds better trade-offs in terms of two-sided fairness when maximizing the lower quantiles of user utility while minimizing the Gini index of item exposure.

In the remainder of the paper, we first describe our recommendation framework. We then present the family of generalized Gini welfare functions and its relationship to previously proposed fairness criteria in ranking. In Sec.~\ref{sec:moreau-ggf} we provide the details of our algorithm and the convergence guarantees. Our experimental results are reported in Sec.~\ref{sec:xps}, and an extension to reciprocal recommendation problems is discussed in Sec.~\ref{sec:reciprocal}. We position our approach with respect to the related work in Sec.~\ref{sec:related}, and Sec.~\ref{sec:conclu} concludes the paper and discusses the limitations of our work.


\section{Fair ranking with Generalized Gini} \label{sec:framework} 

\subsection{Recommendation framework}\label{sec:framework}

We consider a recommendation scenario with $n$ users, and $m$ items, and $K$ recommendation slots. $\mu_{ij} \in [0,1]$ denotes the value of item $j$ for user $i$ (e.g, a ``liking'' probability), and we assume the values $\mu$ are given as input to the system. 
The goal of the system is to produce a ranked list of items for each of the $n$ users. 
Following previous work on fair rankings \citep[e.g.][]{singh2018fairness}, we consider randomized rankings because they enable the use of convex optimization techniques to generate the recommendations, which would otherwise involve an intractable combinatorial optimization problem in the space of all users' rankings. A \emph{randomized ranking} for user $i$ is represented by a bistochastic matrix $\rrk_i \in \reals^{m\times m},$ where $\rrk_{ijk}$ is the probability that item $j$ is recommended to user $i$ at position $k.$ The recommender system is characterized by a \emph{ranking policy} $\rrk = (\rrk_i)_{i=1}^n$. We denote the convex set of ranking policies by $\rrkS$.

We use the term \emph{utility} in its broad sense in cardinal welfare economics as a ``\emph{measurement of the higher-order characteristic that is relevant to the particular distributive justice problem at hand}'' \citep{moulin2003fair}. Similarly to \cite{patro2020fairrec,wang2021user,do2021two}, we define the utility of a user as the ranking performance, and the utility of an item as its average exposure to users, which are formalized in \eqref{eq:utilities} below. Utilities are defined according to the position-based model \citep{biega2018equity,morik2020controlling,do2021two} with weights $\ew \in \reals_+^m$. The weight $\ew_k$ is the probability that a user examines the item at position $k$, and we assume that the weights are non-increasing. Since there are $K$ recommendation slots, we have $\ew_1 \geq \ldots \geq \ew_K$ and $\ew_k = 0$ for any $k > K$. The user and item utilities are then:
\begin{align}\label{eq:utilities}
    \hspace{-1.5em}\resizebox{0.92\linewidth}{!}{$\displaystyle\textit{User utility: }\, u_i(P) = \sum_{j=1}^m \muij \rrk_{ij}^\top \ew \quad\quad\textit{Item exposure: }\, v_j(P) = \sum_{i=1}^n\rrk_{ij}^\top \ew.$}
\end{align}
We follow a general framework where the ranking policy $\rrk$ is found by maximizing a global \emph{welfare function} $F(\rrk)$, and the welfare function is a weighted sum of welfare functions for users and items:
\begin{align}\label{eq:welfobj}
    \welfobj(\rrk) = (1-\lambda) \welfobjuser(\bu(\rrk)) + \lambda \welfobjitem(\bv(\rrk)),
\end{align}
where $\welfobjuser:\Re^n\rightarrow\Re$ and $\welfobjitem:\Re^m\rightarrow\Re$ respectively aggregate the utilities of users and item exposures and $\lambda\in[0,1]$ specifies the relative weight of users and items. 

\subsection{Generalized Gini welfare functions} \label{sec:GGF-pres}

In this work, we focus on the case where $\welfobjitem$ and $\welfobjuser$ are based on Generalized Gini welfare Functions (GGFs) \citep{weymark1981generalized}). A GGF $\giniwelf{\bw}:\Re^n\rightarrow \Re$ is a function parameterized by a vector $\bw \in \reals^n$ of non-increasing positive weights such that $\bw_1=1\geq\ldots\geq \bw_n\geq0$, and defined by a weighted sum of its sorted inputs, which is also called an ordered weighted averaging operator (OWA) \citep{yager1988ordered}. Formally, let $\bx\in\Re^n$ be a utility vector and denote by $\bx\sort$ the values of $\bx$ sorted in increasing order, i.e., $\bx_1\sort \leq ...\leq \bx_n\sort$. Then:
\begin{align}\label{eq:def:ggf}
    \giniwelf{\bw}(\bx) 
    = \sum_{i=1}^{n}\bw_i \bx_i\sort.
\end{align}

Let $\wS_n = \{\bw \in \reals^n:  \bw_1=1 \geq \ldots \geq \bw_n \geq 0\}$ be the set of admissible weights of GGFs. Given $\ubw\in\wS_n$, $\ibw\in\wS_m$ and $\lambda \in(0,1)$, we define the \emph{two-sided GGF} as the welfare function \eqref{eq:welfobj} with $\welfobjuser = \giniwelf{\ubw}$ and $\welfobjitem=\giniwelf{\ibw}$:
\begin{align}\label{eq:def-2sidedggf}
    \welfobj_{\lambda, \ubw, \ibw}(\rrk) = (1-\lambda)\giniwelf{\ubw}\big(\bu(\rrk)\big) + \lambda \giniwelf{\ibw}\big((\bv(\rrk)\big).
\end{align}

With non-increasing, non-negative weights $\bw$, OWA operators are concave \citep{yager1988ordered}. 
The maximization of $\welfobj_{\lambda, \ubw, \ibw}(\rrk)$ \eqref{eq:def-2sidedggf} is thus a convex optimization problem (maximization of a concave function over the convex set of ranking policies). GGFs address fairness from the point of view of distributive justice in welfare economics \citep{moulin2003fair}, because they assign more weight to the portions of the population that have the least utility. Compared to a standard average, a GGF thus promotes more equality between individuals.

\paragraph{Relationship to the Gini index}
GGFs are welfare functions so they follow the convention that they should be maximized. Moreover, if $\bw_i>0$ for all $i$,  $\giniwelf{\bw}$ is increasing with respect to every individual utilities, which ensures that maximizers of GGFs are Pareto-optimal \citep{moulin2003fair}. The Gini index of $\bx$, denoted $\gini(\bx)$ is associated to the GGF $g_{\bw}(\bx)$ with $\bw_i=\nicefrac{(n-i+1)}{n}$ \citep[for formulas of Gini index, see][]{yitzhaki2013more}:
\begin{align}\label{eq:def:giniindex}
    \gini(\bx)
    &= 1-\frac{2}{\norm{\bx}_1}\sum_{i=1}^n \frac{n-i+1}{n} \bx_i\sort\\
    &= \frac{1}{n^2\bxm}\sum_{i=1}^n\sum_{j=1}^n |\bx_i-\bx_j|
    && \text{with~} \bxm=\frac{1}{n}\sum_{i=1}^n \bx_i.
\end{align}
The second equality gives a more intuitive formula as a normalized average of absolute pairwise differences. The Gini index is an inequality measure, and therefore should be minimized, but, more importantly, it is normalized by the sum of utilities $\norm{\bx}_1$, which means that in general minimizing the Gini index does not yield Pareto-optimal solutions. The importance of this normalization is discussed by e.g., \citet{atkinson1970measurement}, and by \citep{do2021two} in the context of fairness in rankings. Yet, when $\bx$ is a vector of item exposures $\bx=\bv(\rrk)$, the normalization is not important because the total exposure is constant. It is then equivalent to minimize the Gini index of item exposures or to maximize its associated GGF.

\paragraph{Multi-objective optimization of Lorenz curves} An alternative formula for $\giniwelf{\bw}(\bx)$ is based on the generalized Lorenz curve\footnote{Lorenz curves are normalized so that the last value is $1$, while generalized Lorenz curves are not normalized.} \citep{shorrocks1983ranking} of $\bx$, which is denoted $\bxcum$ and is defined as the vector of cumulative sums of sorted utilities:
\begin{align}
    \label{eq:def:ggfaslorenz}
    \resizebox{0.92\linewidth}{!}{$\displaystyle \giniwelf{\bw}(\bx) 
    = \sum_{i=1}^n \bw_i' \bxcum_i
    \text{~where~} \bw_i'=\bw_i-\bw_{i+1} \text{~and~} \bxcum_i=\bx_1\sort+\ldots+\bx_i\sort.$}
\end{align}
We used the convention $\bw_{n+1}=0$. Notice that since the weights $\bw$ are non-increasing,  we have that $\bw_{i}'\geq 0$. Thus, family of admissible OWA weights $\bw$ yield weights $\bw'$ that are non-negative and sum to $1$. This formula offers the interpretation of GGFs as positively weighted averages of points of the generalized Lorenz curves. Every GGF thus corresponds to a scalarization of the multi-objective problem of maximizing every point of the generalized Lorenz curve \citep{geoffrion1968proper,miettinen2012nonlinear}. We get back to this interpretation in the next subsections.

\subsection{GGFs for fairness in rankings} \label{sec:fairnesstasks}

To give concrete examples of the relevance of GGFs for fairness in rankings, we provide here two fairness evaluation protocols that have been previously proposed and fall under the scope of maximizing of GGFs as in Eq. \eqref{eq:def-2sidedggf}.

\paragraph{Trade-offs between user utility and inequality in item exposure} The first task consists in mitigating inequalities of exposure between (groups of) items, and appears in many studies \citep{singh2018fairness,zehlike2020reducing,wu2021tfrom}. This leads to a trade-off between the total utility of users and inequality among items, and such inequalities are usually measured by the Gini index (as in \citep{morik2020controlling,biega2018equity}). Removing the dependency on $\rrk$ to lighten the notation, a natural formulation of this trade-off uses the two-sided GGF \eqref{eq:def-2sidedggf} by setting $\ubw = (1,\ldots,1)$ and $\ibw = \left(\frac{m-j+1}{m}\right)_{j=1}^m$, which yields:
\begin{align}\label{eq:def:eqexposure}
    \displaystyle\welfobjuser(\bu)=\frac{1}{n}\sum_{i=1}^n\bu_i
    \quad\quad
    \welfobjitem(\bv)
    =\sum_{j=1}^m \frac{m-j+1}{m} \bv_j\sort.
\end{align}
As stated in the previous section, for item exposure, maximizing $\welfobjitem$ is equivalent to minimizing the Gini index. The Gini index for $\welfobjitem$ has been routinely used for \emph{evaluating} inequality in item exposure \citep{morik2020controlling,biega2018equity} but there is no algorithm to optimize general trade-offs between user utility and the Gini index of exposure. 
\citet{morik2020controlling} use the Gini index of exposures in the context of dynamic ranking (with the absolute pairwise differences formula \eqref{eq:def:giniindex}), where their algorithm is shown to asymptotically drive $\welfobjitem(\bv)$ to $0$, equivalent to $\lambda\to 1$ in \eqref{eq:welfobj}. However, their algorithm cannot be used to converge to the optimal rankings for other values of $\lambda$. \citet{do2021two} use as baseline a variant using the standard deviation of exposures instead of absolute pairwise difference because it is easier to optimize (it is smooth except on $0$). In  contrast, our approach allows for the direct optimization of the welfare function \eqref{eq:welfobj} with this instantiation of $\welfobjitem$ given by eq. \eqref{eq:def:eqexposure}.

Several authors \citep{morik2020controlling,biega2018equity} used \emph{merit-weighted} exposure\footnote{also called \emph{``equity of attention''} \citep{biega2018equity}, \emph{``disparate treatment''} \citep{singh2018fairness}} $\bv'_j(\rrk) = \bv(\rrk)/\overline{\mu}_j$ where $\overline{\mu}_j=\frac{1}{n}\sum_{i=1}^n\muij$ is the average value of item $j$ across users, rather than the exposure itself. We keep the non-weighted exposure to simplify the exposition, but our method straightforwardly applies to merit-weighted exposure. Note however that the sum of weighted exposures is not constant, so using \eqref{eq:def:eqexposure} with merit-weighted exposures is not strictly equivalent to minimizing the Gini index.

\paragraph{Two-sided fairness} \citet{do2021two} propose to add a user-side fairness criterion to the trade-off above, to ensure that worse-off users do not bear the cost of reducing exposure inequalities among items. Their evaluation involves multi-dimensional trade-offs between specific points of the generalized Lorenz curve. Using the formulation \eqref{eq:def:ggfaslorenz} of GGFs, trade-offs between maximizing the cumulative utility at a specific quantile $q$ of users and total utility can be formulated using a parameter $\omega\in[0,1]$ as follows:
\begin{align}\label{eq:quantile}
    \welfobjuser(\bu) = \sum_{i=1}^n \bw_i' \bucum_i \quad \text{with~} \bw_{\lfloor qn\rfloor}'=\omega \text{~and~} \bw_n'=1-\omega,
\end{align}
where all other values of $\bw_i'=0$. In our experiments, we combine this $\welfobjuser$ with the Gini index for $\welfobjitem$ for two-sided fairness.



\subsection{Generating all Lorenz efficient solutions}

In welfare economics, the fundamental property of concave welfare functions is that they are monotonic with respect to the dominance of generalized Lorenz curves \citep{atkinson1970measurement,shorrocks1983ranking,moulin2003fair}, because this garantees that maximizing a welfare function performs an optimal redistribution from the better-off to the worse-off at every level of average utility. In the context of two-sided fairness in rankings, \citet{do2021two} formalize their fairness criterion by stating that a ranking policy is fair as long as the generalized Lorenz curves of users and items are not jointly dominated. In this section, we show that the family of GGFs $\welfobj_{\lambda, \ubw, \ibw}(\rrk)$ \eqref{eq:def-2sidedggf} allows to generate \emph{every} ranking policy that are fair under this definition, and \emph{only} those. The result follows from standard results of convex multi-objective optimization \citep{geoffrion1968proper,miettinen2012nonlinear}. We give here the formal statements for exhaustivity.

Let $\bx$ and $\bx'$ two vectors in $\Re_+^n$. We say that $\bx$ weakly-Lorenz-dominates $\bx'$, denoted $\bx\lorenzw\bx'$, when the generalized Lorenz curve of $\bx$ is always at least equal to that of $\bx'$, i.e., $\bx\lorenzw\bx' \iff \forall i, \bxcum_i \geq \bxcum_i'$. We say that $\bx$ Lorenz-dominates $\bx'$, denoted $\bx\lorenz\bx'$ if $\bx\lorenzw\bx'$ and $\bx\neq\bx'$, i.e., if the generalized Lorenz curve of $\bx$ is strictly larger than that of $\bx'$ on at least one point. The criterion that generalized Lorenz curves of users and items are not jointly-dominated is captured by the notion of \emph{Lorenz-efficiency}:
\begin{definition}[\citet{do2021two}]
A ranking policy $\rrk \in \rrkS$ is \emph{Lorenz-efficient} if there is no $\rrk' \in \rrkS$ such that either [$\bu(\rrk') \lorenzw\bu(\rrk)$ and $\bv(\rrk') \lorenz\bv(\rrk)$] or [$\bv(\rrk') \lorenzw\bv(\rrk)$ and $\bu(\rrk') \lorenz\bu(\rrk)$].
\end{definition}


We now present the main result of this section:
\begin{proposition}\label{lem:all-lorenz} Let $\Theta=(0,1)\times\wS_n\times\wS_m$.
    \begin{enumerate}
        \item Let $(\lambda, \ubw, \ibw)\in\Theta$, where $\ubw$ and $\ibw$ have strictly decreasing weights, and $\rrk^* \in \argmax_{\rrk\in\rrkS}\welfobj_{\lambda, \ubw, \ibw}(\rrk)$. 
        
        Then $\rrk^*$ is Lorenz-efficient. 
        \item If $\rrk$ is Lorenz efficient, then there exists $(\iw,\ubw,\ibw) \in \Theta$ such that $\rrk\in\argmax_{\rrk \in \rrkS} \welfobj_{\lambda, \ubw, \ibw}(\rrk)$.
    \end{enumerate}
\end{proposition}
\begin{proof}
The proof uses standard results on convex multi-objective optimization from \citep{geoffrion1968proper,miettinen2012nonlinear}. Written in the form \eqref{eq:def:ggfaslorenz}, the GGFs corresponds to the scalarization of the multi-objective problem of jointly maximizing the generalized Lorenz curves of users and items, which is a problem with $n+m$ objectives. Indeed, each objective function is a point of the generalized Lorenz curve $(\bucum(P),\bvcum(P)).$ Each objective $U_i(P)$ is concave because it corresponds to an OWA operator with non-increasing weights $\xvector{\rho}$ with $\rho_{i'} = \indic{i'\leq i},$ applied to utilities, which are linear functions of the ranking policy. Each objective $V_i(P)$ is similarly concave. Moreover, we are optimizing over the convex set of stochastic ranking policies $\rrkS$. The multi-objective problem is then concave, which means that the maximizers of all weighted sums of the objectives $(\bucum(P),\bvcum(P))$ with strictly positive weights are Pareto-efficient. Reciprocally every Pareto-efficient solution is a solution of a non-negative weighted sum of the objectives $(\bucum(P),\bvcum(P)),$ where the weights sum to $1$ \citep{miettinen2012nonlinear}.

The result follows from the observation that the Lorenz-efficiency of $\rrk$, defined as the Lorenz-efficiency of $(\bu(P),\bv(P)),$ 
is equivalent to the Pareto-efficiency of its joint user-item Lorenz curves $(\bucum(P),\bvcum(P)).$
This is because the Lorenz dominance relation between vectors $\bx,\bx'$ is defined as Pareto dominance in the space of their generalized Lorenz curves $\xvector{X},\xvector{X}'.$

\end{proof}

\paragraph{Additive welfare functions vs GGFs}
\citet{do2021two} use additive concave welfare functions to generate Lorenz-efficient rankings. Let $\phi(x, \alpha)=x^\alpha$ if $\alpha>0$, $\phi(x, \alpha)=\log(x)$ if $\alpha=0$ and $\phi(x, \alpha)=-x^\alpha$ if $\alpha<0$. \citet{do2021two} use concave welfare functions of the form:
\begin{align}\label{eq:def-qua}
    \welfobjuser(\bu) = \sum_{i=1}^n \phi(\bu_i, \alpha_1) 
    &&
    \welfobjitem(\bv) = \sum_{j=1}^m \phi(\bv_j, \alpha_2)
\end{align}
Where $\alpha_1$ (resp. $\alpha_2$) specifies how much the rankings should redistribute utility to worse-off users (resp. least exposed items). 

Additive separability plays an important role in the literature on inequality measures \citep{dalton1920measurement,atkinson1970measurement,cowell1988inequality}, as well as in the study of welfare functions because additive separability follows from a standard axiomatization \citep{moulin2003fair}. However, this leads to a restricted class of functions, so that varying $\alpha_1, \alpha_2$ and $\lambda$ in \eqref{eq:def-qua} cannot generate all Lorenz-efficient solutions in general. The GGF approach provides a more general device to navigate the set of Lorenz-efficient solutions, with interpretable parameters since they are weights assigned to points of the generalized Lorenz curve.

\section{Optimizing Generalized Gini Welfare} \label{sec:moreau-ggf}

In this section, we provide a scalable method for optimizing two-sided GGFs welfare functions \eqref{eq:def-2sidedggf} $\ggfobj$. The challenge of optimizing GGFs is that they are nondifferentiable since they require sorting utilities. We first describe why existing approaches to optimize GGFs are not suited to ranking in Sec.~\ref{sec:algo-challenge}. We then show how to efficiently compute the gradient of the Moreau envelope of GGFs in Sec.~\ref{sec:algo-gradient} and present the full algorithm in Sec.~\ref{sec:algo-algo}.

\subsection{Challenges} \label{sec:algo-challenge}

In multi-objective optimization, a standard approach to optimizing OWAs is to solve the equivalent linear program derived by \citet{ogryczak2003solving}. Because the utilities depend on 3d-tensors $\rrk \in \rrkS$ in our case, the linear program has $O(n \cdot m^2)$ variables and constraints, which is prohibitively large in practice. 
Another approach consists in using online subgradient descent to optimize GGFs, like \citep{busa2017multi,mehrotra2020bandit}. 
This is not tractable in our case because it requires to project iterates onto the parameter space, which in our case involves costly projections onto the space of ranking policies $\rrkS.$ 
On the other hand, the Frank-Wolfe algorithm \citep{frank1956algorithm} was shown to provide a computationally efficient and provably convergent method to optimize over $\rrkS$ \citep{do2021two}. However, it only applies to smooth functions, and  
Frank-Wolfe with subgradients may not converge to an optimal solution \citep{nesterov2018complexity}. 

We turn to Frank-Wolfe variants for nonsmooth objectives, since Frank-Wolfe methods are well-suited to our structured ranking problem \citep{do2021two,jaggi2013revisiting,clarkson2010coresets}. More precisely, following \citep{lan2013complexity,yurtsever2018conditional,2020thekumparampil}, our algorithm uses the Moreau envelope of GGFs for smoothing. The usefulness of this smooth approximation depends on its gradient, which computation is in some cases intractable \citep{chen2012smoothing}. Our main technical contribution is to show that the gradient of the Moreau envelope of GGFs can be computed in $O(n \log n)$ operations.

\subsection{The Moreau envelope of GGFs}\label{sec:algo-gradient}

In the sequel, 
$\norm{\bz}$ 
denotes the $\ell_2$ norm. Moreover,
%
a function $L: \mathcal{X}\subseteq \reals^n \rightarrow \reals$ is $C$-smooth if it is differentiable with $C$-Lipschitz continuous gradients, i.e., if $\forall \bx,\bx' \in \mathcal{X}\,, \norm{\nabla L(\bx) - \nabla L(\bx')} \leq C \norm{\bx - \bx'}.$

\subsubsection{Definition and properties} Let us fix weights $\bw \in \wS$ and focus on maximizing the GGF $g_{\bw}.$ Let $h(\bz) := -g_{\bw}(\bz)$ 
to obtain a convex function (this simplifies the overall discussion). 
The function $h$ is $\|\bw\|$-Lipschitz continuous, but non-smooth. 
%
%
We consider the smooth approximation of $h$ given by its Moreau envelope \citep{parikh2014proximal} 
defined as:
\begin{align}
    h^\beta(\bz) = \min_{\bz' \in \Re^n} h(\bz') + \frac{1}{2 \beta} \norm{\bz - \bz'}^2.
\end{align}
It is known that $h^\beta(\bz) \leq h(\bz) \leq h^\beta(\bz) +\frac{\beta}{2}\norm{\bw}^2$ and that $h^\beta$ is $\frac{1}{\beta}$-smooth  \citep[see e.g.,][]{2020thekumparampil}. The parameter $\beta$ thus controls the trade-off between the smoothness and the quality of the approximation of $h$.

\subsubsection{Efficient computation of the gradient} \label{sec:isotonic}

\begin{algorithm}
    \caption{\label{alg:grad-moreau} Computation of $\proj_{\ph(\protect\tbw)}$}
    \DontPrintSemicolon
     \SetKwInOut{Input}{input}
     \SetKwInOut{Output}{output}
     
     \Input{GGF weights $\bw\in\Re^n$, $\bz\in\Re^n$}
     \Output{Projection of $\bz$ onto the permutahedron $\ph(\tbw)$.}
     $\tbw \gets -(\bw_n,...,\bw_1)$ and  $\sigma \gets \argsort(\bz)~$\;
     $\bx \gets \PAV(\bz_{\sigma} - \tbw)~$\;
     $\by\gets \bz + \xvector{x_{\sigma^{-1}}}$\;
     Return $\by$.
\end{algorithm}

We now present  an efficient procedure to compute the gradient of $\fbeta(\rrk) := \hbeta(\bu(\rrk))$.

Given an integer $n \in \mathbb{N},$ let $\intint{n} := \{1,\ldots,n\}\,$ and let $\permSn$ denotes the set of permutations of $\intint{n}$. For $\bx \in \reals^n$, and $\sigma\in\permSn$, let us denote by $\bxsig=(\bx_{\sigma(1)}, ..., \bx_{\sigma(n)})$. Furthermore, let $\ph(\bx)$ denote the \emph{permutahedron} induced by $\bx$, defined as the convex hull of all permutations of the vector $\bx$: $\,\ph(\bx) = \conv\{\bxsig : \sigma \in \permSn\}.$ Finally, let $\proj_{\mathcal{X}}(\bz):= \argmin\limits_{\bz' \in \mathcal{X}}\norm{\bz - \bz'}^2.$ denote the projection onto a compact convex $\mathcal{X}$. 
The following proposition formulates $\nabla \fbeta$ as a projection onto a permutahedron:
\begin{proposition}\label{prop:grad-moreau}

Let $\tbw = -(w_n,...,w_1)$. Let $\rrk \in \rrkS$.
Then for all $(i,j,k) \in \intint{n} \times \intint{m}^2$, we have:
\begin{align}
\label{eq:grad-moreau} 
    \frac{\partial \fbeta}{\partial P_{ijk}}(P) = \projbetaPi \muij \ew_k && \text{where}\quad  \projbetaP = \proj_{\ph(\tbw)}\left(\frac{\bu(P)}{\beta}\right).
\end{align}
\end{proposition}

\begin{proof} Let $\prox_{\beta h}(\bz) = \argmin\limits_{\bz' \in \Re^n} \beta h(\bz) + \frac{1}{2}\|\bz'-\bz\|^2$ denote the proximal operator of $\beta h$. Denoting by $\bu*$ the adjoint of $\bu$, it is known that $\grad \fbeta(P) = \frac{1}{\beta} \bu^* (\bu(P) - \prox_{\beta h}(\bu(P))$  \citep{parikh2014proximal}.

We first notice that since $\bw$ are non-increasing, the rearrangement inequalities \citep{hardy1952inequalities} gives:
$h(\bz) = -\min\limits_{\sigma \in \permSn} \bwsigt \bz = \max\limits_{\sigma \in \permSn } -\bwsigt \bz.$
Thus, $h$ is the support function of the convex set $\ph(\tbw)$, since:
\begin{align*}
    h(\bz) = \max_{\sigma \in \permSn} -\bwsigt \bz = \sup_{\by \in \ph(\tbw)} \by^\intercal \bz.
\end{align*}
Then the Fenchel conjugate of $h$ is the indicator function of $\ph(\tbw)$, and its proximal is the projection $\proj_{\ph(\tbw)}$ \citep{parikh2014proximal}. By Moreau decomposition, we get $~\prox(\bz) = \bz - \beta \proj_{\ph(\tbw)}\left(\nicefrac{\bz}{\beta}\right),$ and thus:
\begin{equation}
\begin{aligned}\label{eq:yt-proj}
    \grad \fbeta(P) = \bu^* \left(\proj_{\ph(\tbw)}\left(\nicefrac{\bu(P)}{\beta}\right)\right).
\end{aligned}
\end{equation}
The result follows from the definition of $\bu(P) = \left(\sum_{j,k} \muij P_{ijk} \ew_k\right)_{i=1}^n.$
\end{proof}


Overall, computing the gradient of the Moreau envelope boils down to a projection onto the permutahedron $\ph(\tbw)$. This projection was shown by several authors to be reducible to isotonic regression:
\begin{proposition}[Reduction to isotonic regression \citep{negrinho2014orbit,lim2016efficient,blondel2020fast}]\label{prop:isoto} Let $\sigma\in\permSn$ that sorts $\bz$ decreasingly, i.e. $z_{\sigma(1)} \geq \ldots \geq z_{\sigma(n)}.$
Let $\bx$ be a solution to isotonic regression on $\bzsig - \tbw$, i.e.
\begin{align}
    \bx = \argmin_{x'_1 \leq \ldots \leq x'_n} \frac{1}{2}\|\bx' - (\bzsig - \tbw)\|^2
\end{align}
Then we have: $ \proj_{\ph(\tbw)}(\bz) = \bz + \xvector{x_{\sigma^{-1}}}.$
\end{proposition}

Following these works, we use the Pool Adjacent Violators (\PAV) algorithm for isotonic regression, which gives a solution in $O(n)$ iterations given a sorted input \cite{best2000minimizing}. 
The algorithm for computing the projection is summarized in Alg. \ref{alg:grad-moreau} where we use the notation $\argsort(\bz) = \{\sigma\in \permSn: z_{\sigma(1)} \geq\ldots\geq z_{\sigma(n)}\}$ for permutations that sort $\bz\in\Re^n$ in decreasing order. Including the sorting of $\frac{\bu(P)}{\beta}$, it costs $O(n \log n)$ time and $O(n)$ space. 


\begin{remark}\label{rk:blondel}
Our method is related to the differentiable sorting operator of \citet{blondel2020fast}, which uses a regularization term to smooth the linear formulation of sorting. The regularized form can itself be written as a projection to a permutahedron. The problem they address is different since they differentiate the multi-dimensional sort operation, but eventually the techniques are similar because the smoothing is done in a similar way.
\end{remark}

\begin{remark}
We computed the gradient of $f^\beta(\rrk)=h^\beta(\bu(\rrk))$ with user utilities. The gradient of $f^\beta(\rrk)=h^\beta(\bv(\rrk))$ using item exposures is computed similarly: $\frac{\partial \fbeta}{\partial P_{ijk}}(P) = \projbetaPj \ew_k$  with $\projbetaP = \proj_{\ph(\tbw)}\left(\frac{\bv(P)}{\beta}\right)$.
\end{remark}

\subsection{Frank-Wolfe with smoothing}\label{sec:algo-algo}

We return to the optimization of the two-sided GGF objective \eqref{eq:def-2sidedggf}. In this section, we fix the parameters $(\iw,\ubw,\ibw)$ and consider the minimization of $f := -F_{\iw,\ubw,\ibw}$ over $\rrkS$. 
For $\beta > 0$ we denote by $h_1^\beta$ and $h_2^\beta$ the Moreau envelopes of $-\uGGF$ and $-\iGGF$ respectively. The smooth approximation of $f$ is then:
\begin{align}\label{eq:2sided-moreau}
    \fbeta(P) := (1-\iw) h_1^\beta(\bu(P)) + \iw h_2^\beta(\bv(P)).
\end{align}

Our algorithm \fwsmooth (Alg. \ref{alg:nonsmoothfw}) for minimizing $f$ uses the Frank-Wolfe method for nonsmooth optimization from \citet{lan2013complexity}\footnote{\citet{lan2013complexity} uses the smoothing scheme of \citet{nesterov2005smooth} which is in fact equal to the Moreau envelope (see \citep[Sec. 4.3]{beck2012smoothing}).}. Given a sequence $(\beta_t)_{t\geq 1}$ of positive values decreasing to $0$, the algorithm constructs iterates $\rrkt$ by applying Frank-Wolfe updates to $\fbetat$ at each iteration $t$. More precisely, \fwsmooth finds an update direction with respect
to $\grad \fbetat$ by computing:
\begin{align}\label{eq:fw-step}
    &\Qt = \argmin_{\rrk\in\rrkS} \xdp{\rrk}{\grad \fbetat(\rrktmo)}.
\end{align}
The update rule is $P^{(t)} = \rrktmo + \frac{2}{t+2} \big(Q^{(t)} - \rrktmo\big)$.

Before giving the details of the computation of \eqref{eq:fw-step}, we note that applying the convergence result of \citet{lan2013complexity}, and denoting $D_{\rrkS} =  \max\limits_{\rrk, \rrk' \in \rrkS}\|\rrk - \rrk'\|$ the diameter of $\rrkS$, we obtain\footnote{In more details, the convergence guarantee of \citet{lan2013complexity} uses the operator norm of $\bu$ and $\bv$, which we bound as follows: $\norm{u(P)}^2 \leq \sum_{i} \sum_{j,k} (\muij P_{ijk} \ew_k)^2 \leq \ew_1^2 \norm{P}^2$, because $\muij \in [0,1]$ $\ew_k \in [0,\ew_1],$ and similarly $\norm{\bv(P)}^2 \leq \ew_1^2 \norm{P}^2$. }:
\begin{proposition}[Th. 4, \citep{lan2013complexity}]\label{prop:converge} With $\beta_0 = \frac{2 D_{\rrkS} \ew_1}{\|\bw\|} $ and $\beta_t = \frac{\beta_0}{\sqrt{t}}$, \fwsmooth obtains the following convergence rate: $$f(P^{(T)}) - f(P^*) \leq \frac{2 D_{\rrkS} \ew_1 \|\bw\|}{\sqrt{T}}.$$
\end{proposition}

\begin{algorithm}
    \caption{\label{alg:nonsmoothfw}\fwsmooth. Alg. \ref{alg:grad-moreau} is used for $\protect\projbetaP^1$ and $\protect\projbetaP^2$.}
    \DontPrintSemicolon
     \SetKwInOut{Input}{input}\SetKwInOut{Output}{output}
     
     \Input{values $(\mu_{ij})$, $\#$ of iterations $T$, smoothing seq. $(\beta_t)_{t}$}
     \Output{ranking policy $\rrk^{(T)}$}
     
     Initialize $\rrk^{(0)}$  such that $\rrk_i^{(0)}$ sorts $\mui$ in decreasing order\;
     \For{t=1, \ldots, T}{
     Let \resizebox{0.7\linewidth}{!}{$\displaystyle\projbetaP^1 = \proj_{\ph(\utbw)}\left(\frac{\bu(\rrktmo)}{\beta_t}\right)$ and $\displaystyle\projbetaP^2=\proj_{\ph(\itbw)}\left(\frac{\bv(\rrktmo)}{\beta_t}\right)$}\;
     \For{i=1, \ldots, \n}{
          $\tempmuij =  (1 - \iw) \, \projbetaP_i^1\muij+\iw\projbetaPj^2$\;
          $~~\tilde{\sigma}_i \gets \topk( -\tilde{\mu}_i )$ \tcp{ Update direction \eqref{eq:fw-step}}
          }
     Let $\Qt \in \rrkS$ such that $\Qt_i$ represents  $\tilde{\sigma}_i$\;
     $\rrkt \gets (1- \frac{2}{t+2})\rrktmo +  \frac{2}{t+2}\Qt\,$.
     }
     Return $\rrk^{(T)}$.
\end{algorithm}

\paragraph{Efficient computation of the update direction} For smooth welfare functions of user utilities and item exposures, the update direction \eqref{eq:fw-step} can be computed with only one top-$K$ sorting operation per user \citep{do2021two}.
In our case, the update is given by the following result, where \resizebox{\linewidth}{!}{$\topk(\bz) = \{ \sigma \in \permSn: z_{\sigma(1)}\geq \ldots \geq z_{\sigma(K)} \text{ and } \forall k\geq K, z_{\sigma(K)} \geq z_{\sigma(k)}\}$} is the set of permutations that sort the $k$ largest elements in $\bz$.

\begin{proposition}\label{prop:fwgrad} Let $\tempmu$ defined by $\tempmuij =  (1 - \iw) \, \projbetaP_i^1\muij+\iw\projbetaPj^2\,$ where $\projbetaP^1 = \proj_{\ph(\utbw)}\left(\nicefrac{\bu(\rrktmo)}{\beta_t}\right)$ and $\projbetaP^2=\proj_{\ph(\itbw)}\left(\nicefrac{\bv(\rrktmo)}{\beta_t}\right)$.

For all $i \in \intint{n},$ let $\tilde{\sigma}_i \in \topk(-\tilde{\mu}_i)$ and $\Qt_i$ a permutation matrix representing $\tilde{\sigma}_i$.
Then $\Qt \in \argmin\limits_{\rrk \in \rrkS} \xdp{\rrk}{\grad \fbetat(\rrktmo)}.$
\end{proposition}

\begin{proof}
Using the expression of the gradient of the Moreau envelope derived in Proposition \ref{prop:grad-moreau}, eq. \eqref{eq:grad-moreau}, we have:\begin{align*}
        \resizebox{\linewidth}{!}{$\displaystyle\frac{\partial \fbetat}{\partial P_{ijk}}(\rrktmo)  = (1 - \iw)  \frac{\partial }{\partial P_{ijk}}(h_1^{\beta_t}(\bu(\rrktmo)) + \iw  \frac{\partial }{\partial P_{ijk}}(h_2^{\beta_t}(\bv(\rrktmo))$}   
\end{align*}
And thus $ \frac{\partial \fbetat}{\partial P_{ijk}}(\rrktmo) = \tempmuij \times \ew_k$.
The result then follows from \citep[Lem. 3]{do2021two} and is a consequence of the rearrangement inequality \cite{hardy1952inequalities}: $\Qt$ is obtained by sorting $\tempmuij$ in increasing order, or equivalently, by sorting $-\tempmuij$ in decreasing order.
\end{proof}

Since the computation of the gradient of Moreau envelopes costs $O(n\ln n+n\ln m)$ operations using Alg.~\ref{alg:grad-moreau}, then by Prop. \ref{prop:fwgrad} at each iteration, the cost of the algorithm is dominated by the top-K sort per user, each of which has amortized complexity of $O(m+K\ln K)$: 
\begin{proposition}\label{prop:complexity} Each iteration costs $O(nm + nK\ln K)$ operations.
The total amount of storage required is $O(nK T)$.
\end{proposition}
In conclusion, \fwsmooth has a cost per iteration similar to the standard Frank-Wolfe algorithm for ranking with smooth objective functions. The cost of the non-smoothness of the objective function is a convergence rate of $1/\sqrt{T}$, while the Frank-Wolfe algorithm converges in $O(1/T)$ when the objective is smooth \citep{clarkson2010coresets}.

Moreover, the algorithm produces a \emph{sparse representation} of the stochastic ranking policy as a weighted sum of permutation matrices. In other words, this gives us a Birkhoff-von-Neumann decomposition \cite{birkhoff1940lattice} of the bistochastic matrices \emph{for free}, avoiding the overhead of an additional decomposition algorithm as in existing works on fair ranking \citep{singh2018fairness,wang2021user,su2021optimizing}. 

\section{Experiments} \label{sec:xps}

We first present our experimental setting for recommendation of music and movies, together with the fairness criteria we explore and the baselines we consider. These fairness criteria have been chosen because they were used in the evaluation of prior work, and they exactly correspond to the optimization of a GGF. We thus expect our two-sided GGF $\ggfobj$ to fare better than the baselines, because they allow for the optimization of the exact evaluation criterion. We provide experimental results that demonstrate this claim in Sec. \ref{sec:xp-onesided-tradeoff}. Note that the GGFs are extremely flexible as we discussed in Sec.~\ref{sec:framework}, so our experiments can only show a few illustrative examples of fairness criteria that can be defined with GGFs. In Sec. \ref{sec:xp-convergence}, we show the usefulness of \fwsmooth compared to the simpler baseline of Frank-Wolfe with subgradients.

\subsection{Experimental setup} \label{sec:xp-setup}

Our experiments are implemented in Python 3.9 using PyTorch\footnote{\url{http://pytorch.org}}. For the \PAV algorithm, we use the implementation of Scikit-Learn.\footnote{\url{https://github.com/scikit-learn/scikit-learn/blob/main/sklearn/isotonic.py}}

\subsubsection{Data and evaluation protocol} 

We present experiments on two recommendation tasks, following the protocols of \citep{do2021two,patro2020fairrec}. First, we address music recommendation with \textbf{Lastfm-2k} from \citet{Cantador:RecSys2011} which contains real listening counts of $2k$
users for $19k$ artists on the online music service Last.fm\footnote{\url{https://www.last.fm/}}. We filter the $2,500$ items having the most listeners. In order to show how the algorithm scales, we also consider the \textbf{MovieLens-20m} dataset \cite{harper2015movielens}, which contains ratings in $[0.5,5]$ of movies by users, and we select the top $15,000$ users and items with the most interactions. 

We use an evaluation protocol similar to \cite{patro2020fairrec,do2021two,wang2021user}. For each dataset, a full user-item preference matrix $(\mu_{i,j})_{i,j}$ is obtained by standard matrix factorization algorithms\footnote{Using the Python library Implicit: \url{https://github.com/benfred/implicit} (MIT License).} from the incomplete interaction matrix, following the protocol of \citep{do2021two}. Rankings are inferred from these estimated preferences. The exposure weights $\ew$ are the standard weights of the \emph{discounted cumulative gain} (DCG) (also used in e.g., \cite{singh2018fairness,biega2018equity,morik2020controlling}): $\forall k \in \intint{K}, \ew_k = \frac{1}{\log_2(1+k)}.$ 

The generated $\muij$ are used as ground truth to evaluate rankings, in order to decouple the fairness evaluation of the ranking algorithms from the evaluation of biases in preference estimates (which are not addressed in the paper). The results are the average of three repetitions of the experiments over different random train/valid/test splits used to generate the $\mu_{ij}$.

\subsubsection{Fairness criteria} \label{sec:xp-fairnesscriteria}

We remind two fairness tasks studied in the ranking literature and presented in Section \ref{sec:fairnesstasks}, and describe existing approaches proposed to address them, which we consider as baselines for comparison with our two-sided GGF \eqref{eq:def-2sidedggf} $\ggfobj$.

\paragraph{Task 1: Trade-offs between user utility and inequality between items} We use the two-sided GGF $\ggfobj$ instantiated as in Eq. \eqref{eq:def:eqexposure}, i.e., with $\ubw = (1, \ldots, 1)$ and $\ibw_j = \frac{m - j + 1}{m}$. This corresponds to a trade-off function between the sum of user utilities and a GGF for items with the Gini index weights, where the trade-off is controlled by varying $\lambda \in (0,1).$ We remind though that unlike the standard Gini index, the GGF is un-normalized (see eq. \eqref{eq:def:giniindex}, Sec \ref{sec:GGF-pres}). 

We use three baselines for this task. 

First, since the Gini index is non-differentiable, \citep{do2021two} proposed a differentiable surrogate using the standard deviation (std) instead, which we refer to as \equalexpo:
\begin{align}\label{eq:std-equalexposure}
    F^\util(P) = \sum_{i=1}^n u_i(P) -  \frac{\lambda}{m}\sqrt{\sum_{j=1}^m\left(v_j(P) - \frac{1}{m} \sum_{j'=1}^m v_{j'}(P)\right)^2}
\end{align}

Second, \citet{patro2020fairrec} address the trade-off of Task 1, since they compare various recommendation strategies based on the utility of users and the Lorenz curves of items (see \cite[Fig. 1]{patro2020fairrec}), recalling that the standard Gini index is often defined as $1-2A$ where $A$ is the area under the Lorenz curve \citep{yitzhaki2013more}. Their fairness constraints are slightly different though, as their algorithm \patro\footnote{\citep{patro2020fairrec} consider unordered recommendation lists with a uniform attention model. We transform them into ordered lists using the order output by \patro, and adapt the item-side criterion of minimal exposure to the position-based model.}  guarantees envy-freeness for users, and a minimum exposure of $\frac{\lambda n\norm{\ew}}{m}$ for every item, where $\lambda$ is the user-item tradeoff parameter.

Finally, we use the additive welfare function \eqref{eq:def-qua} (refered to as \welf) with the recommended values $\alpha_1\in\{-2,0,1\}$ and $\alpha_2=0$ \citep{do2021two}, and varying $\lambda\in(0,1)$ as third baseline. We only report the result of $\alpha_1=1$ since it obtained overall better performances on this task.

\paragraph{Task 2: Two-sided fairness} We consider trade-offs between the cumulative utility of the $q$ fraction of worst-off users, where $q 
\in \{0.25, 0.5\},$ and inequality between items measured by the Gini index, as in \citep{do2021two}. For this task, we instantiate the two-sided GGF $\ggfobj$ as follows: the GGF for users is given by Eq. \eqref{eq:quantile} with parameters $(q, \omega)$ in $\{0.25, 0.5\} \times \{0.25, 0.5, 1\}$, and the GGF for items uses the Gini index weights $w_j = \frac{m-j+1}{m}.$ We generate trade-offs between user fairness and item fairness by varying $\lambda \in (0,1).$ 

The baseline approach for this task is \welf, the additive welfare function \eqref{eq:def-qua}, still with the recommended values $\alpha_1\in\{-2,0,1\}$ and $\alpha_2=0$ and varying $\lambda\in(0,1)$. We only report the results of $\alpha_1=-2$ as they obtained the best performances on this task.

\subsection{Results} \label{sec:xp-onesided-tradeoff}

\begin{figure*}
\centering    
\subfigure[Total user utility  vs. inequality between items.   GGF is instantiated as \eqref{eq:def:eqexposure}.]{\label{fig:lastfm-total}\includegraphics[width=0.23\linewidth]{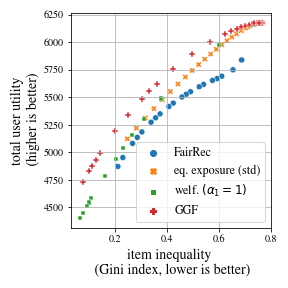}}
\subfigure[Utility of the 25\% worse-off users  vs. inequality between items.   GGF is instantiated as \eqref{eq:quantile} with parameters $(q,\omega)$.]{\label{fig:lastfm-user25}\includegraphics[width=0.23\linewidth]{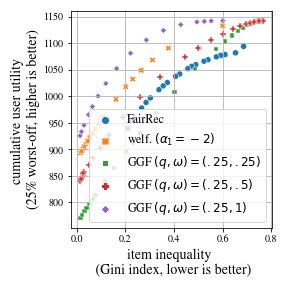}}
\subfigure[Utility of the 50\% worse-off users  vs. inequality between items.   GGF is instantiated as \eqref{eq:quantile} with parameters $(q,\omega)$.]{\label{fig:lastfm-user50}\includegraphics[width=0.23\linewidth]{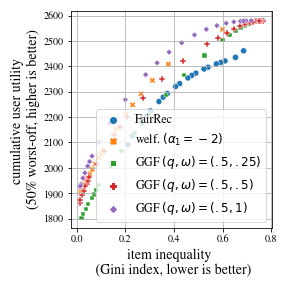}}
\subfigure[Optimization of GGF \eqref{eq:def:eqexposure} with $\iw=0.5$]{\label{fig:lastfm-gini-cvg}\includegraphics[width=0.23\linewidth]{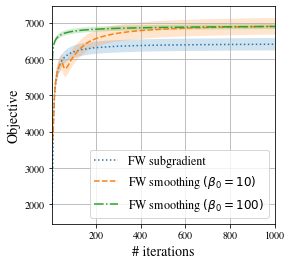}}
\subfigure[Total user utility  vs. inequality between items.   GGF is instantiated as \eqref{eq:def:eqexposure}.]{\label{fig:ml20m-total}\includegraphics[width=0.23\linewidth]{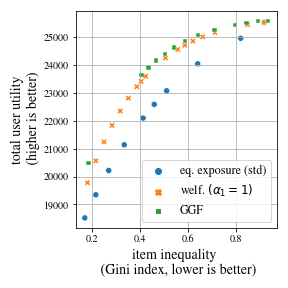}}
\subfigure[Utility of the 25\% worse-off users  vs. inequality between items.   GGF is instantiated as \eqref{eq:quantile} with parameters $(q,\omega)$.]{\label{fig:ml20m-user25}\includegraphics[width=0.23\linewidth]{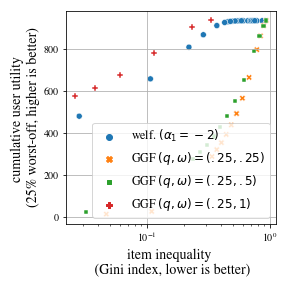}}
\subfigure[Utility of the 50\% worse-off users  vs. inequality between items.   GGF is instantiated as \eqref{eq:quantile} with parameters $(q,\omega)$.]{\label{fig:ml20m-user50}\includegraphics[width=0.23\linewidth]{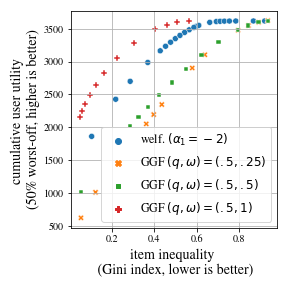}}
\subfigure[Optimization of GGF \eqref{eq:def:eqexposure} with $\iw=0.5$]{\label{fig:ml20m-gini-cvg}\includegraphics[width=0.23\linewidth]{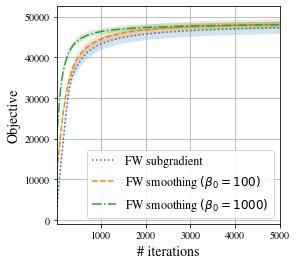}}
\caption{\label{fig:onesided-reco}Summary of the results on \lastfmsmall (top row) and MovieLens (bottom row). (Left 3 columns): Trade-offs achieved by competing methods on various fairness criteria, when varying $\lambda \in [0,1].$ (Right column): Convergence of \fwsubgrad compared to \fwsmooth, for various values of $\beta_0.$ \fwsubgrad is not guaranteed to converge to an optimum.}
\end{figure*}

We now present experiments that illustrate the effectiveness of the two-sided GGF approach on Task 1 and 2.

For each fairness method, Pareto frontiers are generated by varying $\iw$. 
Since \citet{patro2020fairrec}'s algorithm \patro does not scale, we compare to \patro only on \lastfmsmall. 

We optimize $\ggfobj$ using \fwsmooth with $\beta_0=100$ and $T=5k$ for \lastfmsmall, and $\beta_0=1000$ and $T=50k$ for MovieLens. $F^\welf$ and $F^\util$ are optimized with the Frank-Wolfe method of \citep{do2021two} for $T=1k$ and $T=5k$ iterations respectively for \lastfmsmall and MovieLens. This is the number of iterations recommended by \citep{do2021two}, while we need more interactions for \fwsmooth because its convergence is $O(\frac{1}{\sqrt{T}})$ rather than $O(\frac{1}{T})$ because of non-smoothness.

We first focus on \lastfmsmall. On Task 1, \emph{Fig. \ref{fig:lastfm-total}}, the GGF (red $+$ curve) obtains the best trade-off between total utility of users and Gini inequality between items, compared to \patro and \equalexpo. It fares better than \equalexpo (orange $\times$) on this task because \equalexpo reduces inequality between items by minimizing the std of exposures, while GGF with weights $w^2_j = \frac{m-j+1}{m}$ minimizes the Gini index. This shows that Task 1 can be addressed by directly optimizing the Gini index, thanks to GGFs with \fwsmooth.

For Task 2, \emph{Fig. \ref{fig:lastfm-user25} and \ref{fig:lastfm-user50}} depicts the trade-offs achieved between the utility of the $25\%$ / $50\%$ worst-off users and inequality among items. 
First, we observe that the more weight $\omega$ is put on the $q$-\% worst-off in GGF, the higher the curve, which is why we observe the ordering green $\square$ $\prec$ red $+$ $\prec$ purple $\diamond$ for GGF, on both $25\%$ and $50\%$ trade-offs plots. Second, as expected, \welf is outperformed by our two-sided GGF with instantiation \eqref{eq:quantile} 
and $(q,\omega)=(q,1)$ (purple $\diamond$), since it corresponds to the optimal settings for this task. 

Figures \ref{fig:ml20m-total},\ref{fig:ml20m-user25},\ref{fig:ml20m-user50} illustrates the same trade-offs on MovieLens. Results are qualitatively similar: by adequately parameterizing GGFs, we obtain the best guarantees on each fairness task.

Overall, these results show that even though the baseline approaches obtain non-trivial performances on the two fairness tasks above, the direct optimization of the trade-offs involving the Gini index or points of the Lorenz curves, which is possible thanks to our algorithm, yields significant performance gains. Moreover, we reiterate that these two tasks are only examples of fairness criteria that GGFs can formalize, since by varying the weights we can obtain all Lorenz-efficient rankings (Prop. \ref{lem:all-lorenz}). 


\subsection{Convergence diagnostics} \label{sec:xp-convergence}


We now demonstrate the usefulness of \fwsmooth for optimizing GGF objectives, compared to simply using the Frank-Wolfe method of \citep{do2021two} with a subgradient of the GGF (\fwsubgrad). We note that a subgradient of $g_{\bw}(\bx)$ is given by $\xvector{w_{\sigma^{-1}}},$ where 
$\sigma \in \argsort(-\bx)$. 
More precisely, \fwsubgrad is also equivalent to using subgradients of $-\uGGF$ and $-\iGGF$ in Line 3 of Alg. \ref{alg:nonsmoothfw}, instead of $\grad \fbetat (\rrkt), $ ignoring the smoothing parameters $\beta_t$. \fwsubgrad is simpler than \fwsmooth, but it is not guaranteed to converge \citep{nesterov2018complexity}. The goal of this section is to assess whether the smoothing is necessary in practice.

We focus on the two-sided GGF \eqref{eq:def:eqexposure} of Task 1 on \lastfmsmall and MovieLens, using \fwsubgrad and \fwsmooth with different values of $\beta_0$. Figure 
\ref{fig:lastfm-gini-cvg} depicts the objective value as a function of the number of iterations, averaged over three seeds (the colored bands represent the std), on \lastfmsmall. We observe that \fwsubgrad (blue dotted curve) plateaus at a suboptimum. In contrast, \fwsmooth converges (orange dotted and green dash-dot curves), and the convergence is faster for larger $\beta_0$. On MovieLens (Fig \ref{fig:ml20m-gini-cvg}), \fwsubgrad converges to the optimal solution, but it is still slower than \fwsmooth with $\beta_0=1000$.

In conclusion, even though \fwsubgrad reaches the optimal performance on Movielens for this set of parameters, it is still possible that \fwsubgrad plateaus at significantly suboptimal solutions. The use of smoothing is thus not only necessary for theoretical convergence guarantees, but also in practice. In addition, \fwsmooth has comparable computational complexity to \fwsubgrad since the computation cost is dominated by the sort operations in Alg.~\ref{alg:nonsmoothfw}.

\section{Reciprocal recommendation} \label{sec:reciprocal}

\subsection{Extension of the framework and algorithm}

We show that our whole method for fair ranking readily applies to reciprocal recommendation tasks, such as the recommendation of friends or dating partners, or in job search platforms.

\paragraph{Reciprocal recommendation framework} The recommendation framework we discussed thus far depicted \emph{``one-sided'' recommendation}, in the sense that only items are being recommended. In \emph{reciprocal recommendation} problems \citep{palomares2021reciprocal}, users are also items who can be recommended to other users (the item \emph{per se} is the user's profile or CV), and they have preferences over other users. 

In this setting, $n = m$ and $\muij$ denotes the mutual preference value between $i$ and $j$ (e.g., the probability of a ``match'' between $i$ and $j$). Following \citep{do2021two}, we extend our previous framework to reciprocal recommendation by introducing the \emph{two-sided utility} of a user $i$, which sums the utility $\Uover_i(P)$ derived by $i$ from the recommendations it gets, and the utility $\Vover_i(P)$ from being recommended to other users:
\begin{align*}
    &u_i(P) = \Uover_i(P) + \Vover_i(P) \, = \sum_{i,j} (\muij + \muji) P_{ij}^\intercal \ew \\
    &\text{where } ~~\Uover_i(P) = \sum_{j=1}^n \muij P_{ij}^\intercal \ew \quad \text{ and } \quad \Vover_i(P) = \sum_{j=1}^n \muij P_{ji}^\intercal \ew\,.
\end{align*}

\paragraph{Objective and optimization} The two-sided GGF objective \eqref{eq:def-2sidedggf} in reciprocal recommendation simply becomes one GGF of two-sided utilities, and it is specified by a single weighting vector $\bw$: 
\begin{align}\label{eq:ggf-reciprocal}
  \max\limits_{\rrk \in \rrkS}\{F_{\bw}(P) := g_{\bw}(\bu(P))\}.   
\end{align}

The choice of $\bw$ controls the degree of priority to the worse-off in the user population. We show in our experiments in Section \ref{sec:xp-reciprocal} that in reciprocal recommendation too, the GGF objective can be adequately parameterized to address existing fairness criteria.

$F_{\bw}$ can be optimized using our algorithm \fwsmooth. Since there is only one GGF, the subroutine Alg. \ref{alg:grad-moreau} is simply used once per iteration to project onto $\ph(\tbw),$ and obtain $\by = \proj_{\ph(\tbw)}\left(\frac{\bu(\rrk^{(t-1)})}{\beta_t}\right)$ in Line 3 of Algorithm \ref{alg:nonsmoothfw}. Line 5 becomes $\tempmuij = (1-\lambda) y_i \muij + \lambda y_j \muji.$

Our method for fair ranking is thus general enough to address both one-sided and reciprocal recommendation, using the notion of two-sided utility from \citet{do2021two}. 

\subsection{Experiments} \label{sec:xp-reciprocal}

Similarly to our experiments in Sec. \ref{sec:xp-onesided-tradeoff}, the goal of these experiments is to demonstrate that in reciprocal recommendation, GGF can be parameterized to exactly optimize for existing fairness criteria, outperforming previous approaches designed to address them. 

\begin{figure}
\centering    
\subfigure[Total utility vs. inequality. GGF is instantiated as \eqref{eq:reciprocal-gini} and we vary $\lambda.$]{\label{fig:twitter13k-total}\includegraphics[width=0.49\linewidth]{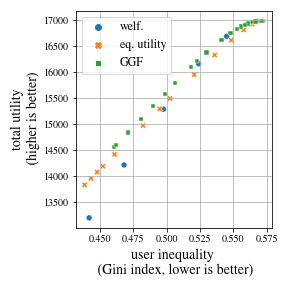}}
\subfigure[Total utility vs. utility of the 25\% worse-off. GGF is instantiated as \eqref{eq:quantile} with $q=0.25$ and varying $\omega.$]{\label{fig:twitter13k-user25}\includegraphics[width=0.49\linewidth]{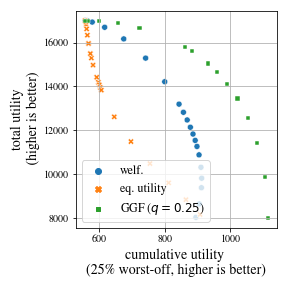}}
\caption{\label{fig:twitter13k-paretofronts}Fairness trade-offs achieved by competing methods on reciprocal recommendation on Twitter. They are generated by varying $\lambda$ for \equalu and GGF, and $\alpha$ for \welf.
}
\end{figure}

\paragraph{Data} We reproduce the experimental setting of \citep{do2021two} who also study fairness in reciprocal recommendation. We simulate a friend recommendation task from the Higgs \textbf{Twitter} dataset \citep{de2013anatomy}, which contains directed follower links on the social network Twitter. We consider a mutual follow as a ``match'', and we keep users having at least 20 matches, resulting in a subset of $13k$ users. We estimate mutual scores $\muij$ (i.e., match probabilities) by matrix factorization.

\subsubsection{Fairness criteria} Similarly to Section \ref{sec:xp-fairnesscriteria}, we state two fairness tasks that previously appeared in the literature, we instantiate the GGF objective $F_{\bw}$ for each task and describe existing baselines.

\paragraph{Task 1: Trade-offs between total utility and inequality of utility among users} Although reciprocal recommendation received less attention in the fairness literature, an existing requirement is to mitigate inequalities in utility between users \citep{jia2018online,basu2020framework}, similarly to the \equalexpo criterion in one-sided recommendation. This leads to a trade-off between the sum of utilities and inequality typically measured by the Gini index. For Task 1, we use the GGF \eqref{eq:ggf-reciprocal} with $\bw_i = (1 - \lambda) + \lambda \cdot \frac{n - i + 1}{n},$ which yields a trade-off function between the sum of utilities and inequality of utilities, and we vary $\lambda$ in $(0,1)$ to generate such trade-offs:
\begin{align}\label{eq:reciprocal-gini}
    F_{\bw}(P) = (1 - \lambda) \sum_{i=1}^n u_i(P) + \lambda \sum_{i=1}^n \frac{n - i + 1}{n} u_i\sort(P).
\end{align}

We use two baselines for this task.

First, similarly to \equalexpo, to bypass the nonsmoothness of the Gini index, \citep{do2021two} optimize a surrogate with std, named \equalu:\begin{align}\label{eq:def-eq-util}
    F^\util(P) = \sum_{i=1}^n u_i(P) -  \frac{\lambda}{n}\sqrt{\sum_{i=1}^n\left(u_i(P) - \frac{1}{n} \sum_{i'=1}^n u_{i'}(P)\right)^2}.
\end{align}

Second, the welfare function \welf \eqref{eq:def-qua} of \citep{do2021two} is used in reciprocal recommendation as a single sum: $F^\welf(P) = \sum_{i=1}^n \phi(u_i(P), \alpha)$ where $\phi$ is defined in Sec. \ref{sec:fairnesstasks}. We study \welf as baseline by varying $\alpha,$ which controls the redistribution of utility in the user population.

\paragraph{Task 2: Trade-offs between total utility and utility of the worse-off} The main task studied by \citep{do2021two} with \welf is to trade-off between the total utility and the cumulative utility of the $q$ fraction of worse-off users. For this task, we instantiate the GGF with \eqref{eq:quantile}, with fixed quantile $q =0.25$ and we vary $\omega$ to generate trade-offs between total utility and cumulative utility of the $25\%$ worst-off.

We compare it to the \welf baseline where $\alpha$ is varied as in \citep{do2021two}.

\subsubsection{Fairness trade-offs results}

\paragraph{Results} We now demonstrate that in reciprocal recommendation too, GGF is the most effective approach in addressing existing fairness criteria. We optimize the GGF $F_{\bw}(P)$ using \fwsmooth with $\beta_0=10$ for $T=50k$ iterations, and optimize $F^\welf$ and $F^\util$ using Frank-Wolfe for $T=5k$ iterations.

Figure \ref{fig:twitter13k-paretofronts} depicts the trade-offs obtained by the competing approaches on the fairness tasks 1 and 2, on the Twitter dataset. Fig. \ref{fig:twitter13k-total} illustrates the superiority of GGF (green $\square$) on Task 1, despite good performance of the baselines \equalu (orange $\times$) and \welf (blue $\circ$). As in one-sided recommendation with \equalexpo, the reason why \equalu achieves slightly worse trade-offs on this fairness task is because it minimizes the std as a surrogate to the Gini index, instead of the Gini index itself as GGF does. For Task 2, on Fig.\ref{fig:twitter13k-user25}, we observe that GGF with parameterization \eqref{eq:quantile} (green $\square$) is the most effective. This is because unlike the \welf approach (blue $\circ$) of \cite{do2021two} who address this fairness task, this form of GGF is exactly designed to optimize for utility quantiles.

\section{Related work}  \label{sec:related}

\paragraph{Algorithmic fairness} Fairness in ranking and recommendation systems is an active area of research. Since recommender systems involve multiple stakeholders \citep{burke2017multisided,abdollahpouri2020multistakeholder}, fairness has been considered from the perspective of both users and item producers. On the user side, a common goal is to prevent disparities in recommendation performance across sensitive groups of users \citep{mehrotra2017auditing,ekstrand2018all}. On the item side, authors aim to prevent winner-take-all effects \citep{abdollahpouri2019unfairness} by redistributing exposure across groups of producers, either towards equal exposure, or equal ratios of exposure to relevance \citep{singh2018fairness,biega2018equity,diaz2020evaluating,kletti2022introducing}, sometimes measured by the classical Gini index \citep{morik2020controlling,wilkie2014best}. 

Some authors consider fairness for both users and items, often by applying existing user or item criteria simultaneously to both sides, such as  \citep{basu2020framework,wu2021tfrom,wang2021user}. \citep{patro2020fairrec,DBLP:journals/corr/abs-2104-14527} instead discuss two-sided fairness with envy-freeness as user-side criterion, while \citep{deldjoo2021flexible} propose to use generalized cross entropy to measure unfairness among sensitive groups of users and items. \citep{wu2021multi} recently considered two-sided fairness in recommendation as a multi-objective problem, where each objective corresponds to a different fairness notion, either for users or items. Similarly, \citet{mehrotra2020bandit} aggregate multiple recommendation objectives using a GGF, in a contextual bandit setting. In their case, the aggregated objectives represent various metrics (e.g., clicks, dwell time) for various stakeholders
. Unlike these two works \citep{wu2021multi,mehrotra2020bandit}, in our case the multiple objectives are the individual utilities of each user and item, and our goal is to be fair towards each entity by redistributing utility. To our knowledge, we are the first to use GGFs as \emph{welfare functions} of users' and items' utilities for two-sided fairness in rankings.

Reciprocal recommender systems received comparatively less attention in the fairness literature, to the exception of \citep{jia2018online,xia2015reciprocal,paraschakis2020matchmaking}. The closest to our work is the additive welfare approach of \citep{do2021two}, which addresses fairness in both one-sided and reciprocal recommendation, and is extensively discussed in the paper, see Sec. \ref{sec:framework}.

In the broader fair machine learning community, several authors advocated for economic concepts \citep{finocchiaro2020bridging}, using inequality indices to quantify and mitigate unfairness \citep{wilkie2015retrievability,speicher2018unified,heidari2018fairness,lazovich2022measuring}, taking an axiomatic perspective \citep{golz2019paradoxes,cousins2021axiomatic,williamson2019fairness} or applying welfare economics principles \citep{hu2020fair,rambachan2020economic}. GGFs, in particular, were recently applied to fair multi-agent reinforcement learning, with multiple reward functions \citep{busa2017multi,siddique2020learning,zimmer2021learning}. These works consider sequential decision-making problems without ranking, and their GGFs aggregate the objectives of a few agents (typically $n < 20$), while in our ranking problem, there are as many objectives as there are users and items.

\paragraph{Nonsmooth convex optimization and differentiable ranking} Our work builds on nonsmooth convex optimization methods \citep{nesterov2005smooth,shamir2013stochastic}, and in particular variants of the Frank-Wolfe algorithm \citep{frank1956algorithm,jaggi2013revisiting} for nonsmooth problems \citep{lan2013complexity,yurtsever2018conditional,ravi2019deterministic,2020thekumparampil}. The recent algorithm of \citep{2020thekumparampil} is a Frank-Wolfe variant which uses the Moreau envelope like us. Its number of first-order calls is optimal, but this is at the cost of a more complex algorithm with inner loops that make it slow in practice. In our case, since the calculation of the gradient is not a bottleneck, we use the simpler algorithm of \citet{lan2013complexity}, which applies Frank-Wolfe to the Moreau envelope of the nonsmooth objective.

Our technical contribution is also related to the literature on differentiable ranking, which includes a large body of work on approximating learning-to-rank metrics \citep{chapelle2010gradient,taylor2008softrank,adams2011ranking}, and recent growing interest in designing smooth ranking modules \citep{grover2019stochastic,cuturi2019differentiable,blondel2020fast} for end-to-end differentiation pipelines. The closest method to ours is the differentiable sorting operator of \citet{blondel2020fast}, which also relies on isotonic regression. The differences between our approaches are explained in Remark \ref{rk:blondel}. 



\section{Conclusion} \label{sec:conclu}

We proposed generalized Gini welfare functions as a flexible method to produce fair rankings. We addressed the challenges of optimizing these welfare functions by leveraging Frank-Wolfe methods for nonsmooth objectives, and demonstrated their efficiency in ranking applications. Our framework and algorithm applies to both usual recommendation of movies or music, and to reciprocal recommendation scenarios, such as dating or hiring.

Generalized Gini welfare functions successfully address a large variety of fairness requirements for ranking algorithms. On the one hand, GGFs are effective in reducing inequalities, since they generalize the Gini index in economics. Optimizing them allows to meet the requirements of equal utility criteria, largely advocated by existing work on fair recommendation \citep{singh2018fairness,basu2020framework,patro2020fairrec,wu2021tfrom}. On the other hand, GGFs effectively increase the utility of the worse-off, which is usually measured by quantile ratios in economics, and has been recently considered as a fairness criterion in ranking \citep{do2021two}.

Our approach is limited to fairness considerations at the stage of inference. It does not address potential biases arising at other parts of the recommendation pipeline, such as in the estimation of preferences. Moreover, we considered a static model, which does not accounts for real-world dynamics, such as responsiveness in two-sided markets \citep{su2021optimizing}, feedback loops in the learning process \citep{bottou2013counterfactual}, and the changing nature of the users' and items' populations \citep{morik2020controlling} and preferences \citep{kalimeris2021preference}. Addressing these limitations, in combination with our method, are interesting directions for future research.

\begin{acks}
We would like to thank Sam Corbett-Davies for his thoughtful feedback on this work.
\end{acks}

\bibliographystyle{ACM-Reference-Format}
\balance
\bibliography{references.bib}



\end{document}